\documentclass[submission,copyright,creativecommons,sharealike]{eptcs}

\providecommand{\event}{INFINITY 10}

\providecommand{\firstpage}{1}

\usepackage[latin1]{inputenc}
\usepackage[english]{babel}

\usepackage{amsthm}
\usepackage{amsmath,amssymb,latexsym,enumerate,array,xspace,mathtools}
\usepackage{pgfmath}
\usepackage{tikz}
\usetikzlibrary{arrows,automata}
\usepackage{enumerate}
\usepackage{stmaryrd}

\newcommand{\laformule}{\varphi_{0}}
\newcommand{\pptit}{\leq}
\newcommand{\pgd}{\geq}
\newcommand{\fifi}{\varphi}

\newcommand{\nit}{\ensuremath{\mathbb{N}}} 
\newcommand{\ratio}{\ensuremath{\mathbb{Q}}} 
\newcommand{\vers}{\rightarrow}
\newcommand{\versdans}[1]{\xrightarrow[#1]{}}
\newcommand{\qqs}{\forall}

\newcommand{\chemin}[2]{\xRightarrow[#2]{#1}}               
                    
\newcommand{\larc}[1]{\xrightarrow{#1}}                    
           
\newcommand{\abs}[1]{\lvert#1\rvert}     

\newcommand{\resp}{{\it resp. }}

\newcommand{\prgs}{{probabilistic regular graphs}\xspace}

\newcommand{\ens}[1]{\left \{#1\right\}} 
\newcommand{\init}[1]{I}  
\newcommand{\fini}[1]{F}  

\newcommand{\injec}{\iota}
\newcommand{\canon}{\mathsf{Can}}
\newcommand{\level}{\mathsf{Lev}}

\newcommand{\alphab}{\Sigma} 

\newcommand{\hrgg}{{\sc hr}-gram\-mar\xspace}
\newcommand{\hrggs}{{\sc hr}-gram\-mars\xspace}
\newcommand{\axio}{Z}

\newcommand{\canbis}{\mathsf{Can}}

\newcommand{\phrgg}{{\sc phr}-gram\-mar\xspace}
\newcommand{\phrggs}{{\sc phr}-gram\-mars\xspace}
\newcommand{\proba}{\mu}
\newcommand{\compte}{\#}
               
\newcommand{\new}{\mbox{{\bf New}}}

\newcommand*{\modality}[2][]{\def\@rgone{#1}%
  \ifx\@rgone\@empty
  \ensuremath{\text{\rmfamily\upshape\bfseries {#2}}}%
  \else
  \ensuremath{\text{\rmfamily\upshape\bfseries {#2}}_{#1}}%
  \fi}
\newcommand*{\dmodality}[2][]{\def\@rgone{#1}%
  \ifx\@rgone\@empty
  \ensuremath{\smash{\widetilde{\text{\rmfamily\upshape\bfseries {#2}}}}%
        {\vphantom{\text{\rmfamily\upshape\bfseries {#2}}}}}%
  \else
  \ensuremath{\smash{\widetilde{\text{\rmfamily\upshape\bfseries {#2}}}_{#1}}%
        {\vphantom{\text{\rmfamily\upshape\bfseries {#2}}_{#1}}}}%
  \fi}
\newcommand*{\Until}[1][]{\ifmmode\,\fi\modality[{#1}]{U}\ifmmode\,\else\expandafter\xspace\fi}

\newcommand*{\G}[1][]{\modality[{#1}]{G}\ifmmode\,\else\expandafter\xspace\fi}
\newcommand*{\F}[1][]{\modality[{#1}]{F}\ifmmode\,\else\expandafter\xspace\fi}
\newcommand*{\Next}[1][]{\modality[{#1}]{X}\ifmmode\,\else\expandafter\xspace\fi}

\newcommand{\prob}{\mathbb{P}} 
\newcommand{\calP}{\mathcal{P}} 
\newcommand{\M}{\mathcal{M}} 
\newcommand{\Succ}{\mathsf{Succ}} 
\newcommand{\calG}{\mathcal{G}} 
\newcommand{\conf}[1]{V_{#1}} 
\newcommand{\sem}[1]{\llbracket #1 \rrbracket} 
\newcommand{\AP}{\mathsf{AP}} 
\newcommand{\Win}{\mathbb{W}} 
\newcommand{\Dec}{\mathbb{D}} 

 \newtheorem{theorem}{\bf Theorem}
 \newtheorem{proposition}[theorem]{\bf Proposition}

\theoremstyle{definition}
\newtheorem{definition}[theorem]{Definition}
\newtheorem{example}[theorem]{Example}
\newtheorem{remark}[theorem]{Remark}

\numberwithin{theorem}{section}
\numberwithin{figure}{section} 

\title{Probabilistic regular graphs}

\author{Nathalie Bertrand
\institute{INRIA Rennes Bretagne Atlantique}
\email{nathalie.bertrand@inria.fr}
\and 
Christophe Morvan
\institute{Universit\'e Paris-Est}
\institute{INRIA Rennes Bretagne Atlantique} 
\email{christophe.morvan@univ-paris-est.fr}
}
\date{\today}

\def\titlerunning{Probabilistic regular graphs}
\def\authorrunning{N. Bertrand and C. Morvan}

\begin{document}

\maketitle

\begin{abstract}
  {\bf Abstract.} Deterministic graph grammars generate regular
  graphs, that form a structural extension of configuration graphs of
  pushdown systems. In this paper, we study a probabilistic extension
  of regular graphs obtained by labelling the terminal arcs of the
  graph grammars by probabilities. Stochastic properties of these
  graphs are expressed using PCTL, a probabilistic extension of
  computation tree logic. We present here an algorithm to perform
  approximate verification of PCTL formulae. Moreover, we prove that
  the exact model-checking problem for PCTL on probabilistic regular
  graphs is undecidable, unless restricting to qualitative
  properties.
  Our results generalise those of~\cite{EKM06}, on probabilistic
  pushdown automata, using similar methods combined with graph
  grammars techniques.
\end{abstract}

\section{Introduction}

Formal methods have proven their importance in the validation of
hardware and software systems. In order to represent real systems more
accurately, several aspects need to be reflected in the
model. Recursion and random events are examples of such extra features
and lead to complex models that incorporate two sources of complexity:
probabilities and infinite state space. For each of these features
independently, verification techniques have been established.

Infinite state systems, on the one hand, cover a large range of
expressive power. Among them pushdown systems offer a simple infinite
framework by extending finite state systems with a stack. Despite the
fact that their configurations graph is infinite, pushdown systems
enjoy several interesting properties. In particular, the reachability
problem is decidable, and the reachability set is effectively
regular~\cite{buchi:regular}. Moreover, monadic second order logic
(MSO)~\cite{muller} is decidable over the graph of configurations for
pushdown automata.
Alternatively, the configurations graphs of pushdown automata can be
generated by deterministic graph grammars, introduced by
Courcelle~\cite{courcelle}. Deterministic graph grammars generate
\emph{regular graphs} which also have decidable MSO~\cite{courcelle},
and which characterise the same structures as pushdown
systems~\cite{caucal01} when restricting to finite degree. 
We advocate that these grammars offer a simple presentation and
emphasize the structural properties of graphs.
Indeed, contrary to pushdown automata, graph grammars are more robust
to transformations. Precisely, many transformations of pushdown
automata affect the configurations graph, and thus its stucture-based
properties. On the contrary, graph grammars allow for transformations
in the representations which preserve the structure. Indeed, most
graph grammar transformations presented in~\cite{caucal:gragra}
preserve, up to isomorphism, the generated graph. Using such
representations thus seems promising in order to express structural
properties of systems.

Probabilistic systems, on the other hand, also raised intensive
research concerning verification, starting with model-checking
algorithms for Markov chains, and Markov decision processes for
various logics. In the last decade, models combining probabilities and
infinite-state spaces have been investigated. Examples of such models
are probabilitic pushdown systems and probabilistic lossy channel
systems. These systems are finitely described and generate infinite
Markov chains on which one can express probabilistic properties, for
example using the probabilistic extension of CTL,
PCTL~\cite{HJ-fac94}.
This logic allows to express, \emph{e.g.}, the probability of
satisfying a given CTL path formula. More generally, PCTL can be seen
as a variant of CTL where the usual forall quantifier is replaced with
a probabilistic comparison to a threshold: the whole state formula is
satisfied if the probability of the set of executions satisfying the
CTL path formula meets the constraint expressed by the threshold. A
restricted fragment of this logic, called \emph{qualitative} PCTL is
obtained when allowing values $0$ and $1$ only for the thresholds. In
constrast, the general case (where threshold values are arbitrary) is
referred to as \emph{quantitative} PCTL.
 The model-checking problem for probabilistic
logics over infinite Markov chains generated by probabilistic lossy
channel systems or probabilistic pushdown automata is a natural and
deeply investigated issue.
Concerning probabilistic pushdown automata, a series of papers
established fundamental model checking
results~\cite{BKS-stacs05,EKM06,Kucera-entcs06,BBHK-lpar08},
some of the most significant ones being the decidability of the model
checking of qualitative PCTL formulae, and the undecidability of the
quantitative version.

In this paper, we consider a probabilistic extension of regular
graphs. To this aim, we define probabilistic graph grammars as graph
grammars where terminal arcs are labelled with
probabilities. Probabilistic graph grammars hence generate
infinite-state Markov chains, and form a natural generalisation of
probabilistic pushdown automata. For these models, we extend the
results of~\cite{EKM06} concerning the model-checking of
PCTL. Precisely, for probabilistic graph grammars we prove the
decidability of the qualitative PCTL model-checking ; we detail how to
approximate the probability of path formula ; and we prove the
undecidability of the exact quantitative PCTL model-checking.

\section{Regular graphs and probabilistic regular graphs}

\subsection{Hypergraphs and graphs}
Let $F$ be a ranked alphabet, and $\rho : F \vers \nit$ its ranking
function that assigns to each element of $F$ its \emph{arity}. We
denote by $F_n$ the set of symbols of arity $n$. Given $V$ an
arbitrary set of vertices, a \emph{hypergraph} $G$ is a subset of
$\cup_{n\pgd 1} F_n V^n$. The vertex set of $G$, denoted $V_G$, is
defined as the set $V_G = \{v\in V\, |\, FV^*vV^* \cap G \neq
\emptyset\}$. In our setting, this set is always countable. An element
of $F_n V^n$ is an {\em hyperarc} of arity $n$, denoted by $f\ v_1\
v_2\ \cdots \ v_n$.

Graphs form a restricted class of hypergraphs where hyperarcs have
arity at most $2$. Precisely, a {\em graph} $G$ over $V$ is a subset
of $F_2VV\cup F_1V$. For $a \in F_2$, and $s,t \in V$, $ast \in G$ is
an \emph{arc} of $G$ with source $s$, target $t$ and label $a$. For $a
\in F_1$ and $s \in V$, if $as$ is an element of $G$, $a$ is referred
to as the \emph{colour} of vertex $s$ (observe that a vertex may have
several colours). $Dom(G)$, $Im(G)$ and $V_G$ denote respectively the
set of sources, targets and vertices of $G$. The {\em in-degree}
(\resp {\em out-degree}) of a vertex $v$ is the number of arc having
source (\resp target) $v$; its {\em degree} is the sum of the in and
out-degrees.  The transition relation underlying $G$ is composed of
transitions $s \xrightarrow{a}_G t$ for $ast \in G$. A \emph{path} in
$G$ is a finite sequence of transitions $v_1 \xrightarrow{a_1} v_2
\cdots \xrightarrow{a_{n-1}} v_n$, also noted $v_1 \stackrel{a_1
  \cdots a_n}{\Rightarrow}_G v_n$.

A {\em graph morphism} from $G$ to $G'$, is a mapping $g : V_G
\rightarrow V_{G'}$ such that for all $u,v \in V_G$, $u\xrightarrow{a}_G
v$ implies $g(u)\xrightarrow{a}_{G'} g(v)$. Such a morphism is an {\em
  isomorphism} if $g$ is a bijection, and its inverse is also a
morphism.

\subsection{Graph grammars}
Graph grammars are a convenient tool to represent graph
transformations. Starting from a hyperarc, the axiom, and using
rewriting rules, these grammars generate families of infinite graphs
that enjoy interesting properties (for example the decidability of MSO
theory, or the fact that they generate context-free languages). Graphs
generated by graph grammars form a slight extension of the graphs of
configurations for pushdown automata, namely such a graph may have
vertices of infinite degree (still there are only finitely many
distinct degrees). A motivation for generating these graphs using
graph grammars rather than pushdown automata is to emphasize the
structural properties of the obtained graphs, since they are defined
up to isomorphism. In particular, stochastic properties of Markov
chains (like probability of a path or a set of paths) are invariant
under graph isomorphism, this justifies the use of structural
characterizations such as graph grammars.

  \begin{definition}\label{def:gragra}
    A {\em hypergraph grammar} (\hrgg for short), is a tuple $\calG =
    (N, T, R, \axio)$, where:
    \begin{itemize} 
    \item $N$ and $T$ are two ranked alphabets of {\em non-terminal}
      and {\em terminal} symbols, respectively;
    \item $\axio \in N$ is a $0$-arity non-terminal, the
      \emph{axiom};
    \item $R$ is a set of \emph{rewriting rules} assigning to each
      non-terminal $A\in N$ a pair $(H_A,\injec_A)$ where $H_A$ is a
      finite hypergraph, and $\injec_A:\{1,\cdots,\rho(A)\}
      \hookrightarrow V_{H_A}$ is an injective mapping associating to
      each position in an hyperedge labelled $A$ a vertex in $H_A$.
    \end{itemize}
  \end{definition}

\begin{example}\label{ex:running}
  Figure~\ref{fig:gramrun} presents an example of a \hrgg.
  Formally, it is defined by $\calG=(\ens{Z}_0\cup\ens{A}_2,
  \ens{V_1,V_2}_1\cup \ens{a,d}_2,
  \ens{(H_Z,\injec_Z),(H_A,\injec_A)}, Z) $.
  Non-terminal $Z$ (\resp $A$) is the only arity $0$ (\resp $2$)
  non-terminal symbol; $\ens{V_1,V_2}$ (\resp $\ens{a,d}$) are the two
  colours (\resp arc-labels); hypergraphs $H_Z$, $H_A$ and injection
  $\injec_A$ are represented in the first part of the figure. For
  simplicity, $\overline{V_1}$ denotes the absence of colour
  $V_1$. The injection $\injec_A$ is used to identify vertices of
  $H_A$ with vertices of an arc labelled $A$ in the rewriting process
  defined later on.
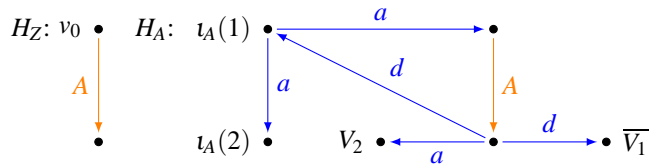
\begin{figure}[htbp]
\begin{center}
  {\small
    \begin{tikzpicture}[scale=.75, auto=left]
      %
      \tikzstyle{every label}=[label distance= 0.5mm];
      \tikzstyle{sommets}=[draw,shape=circle,inner sep=0pt, outer sep = 0pt, fill= black,minimum size= 1mm];
      \tikzstyle{recriture}=[line width = 1pt, arrows = -latex];
      \tikzstyle{arct}=[arrows = -latex, shorten >=2pt, shorten <=2pt,color=blue];
      \tikzstyle{arcnt}=[arrows = -latex, shorten >=2pt, shorten <=2pt,color=orange];
      
      \draw 
      node (HZ) at (-2.2,2) {$H_Z$:}
      node (as) at (-1,2) [sommets, label=left:$v_0$]{}
      node (at) at (-1,0) [sommets]{} 
      node (HA) at (0,2) {$H_A$:}
      ++(2,2) 
      node (s) at +(0,0) [sommets, label=left:$\injec_A(1)$]{}
      node (t) at +(0,-2) [sommets, label=left:$\injec_A(2)$]{} 
      node (Ap) at +(2,-2) [sommets, label=left:$V_2$]{}
      node (ABq) at +(4,-2) [sommets]{}
      node (new1) at +(6,-2)[sommets, label=right:$\overline {V_1}$]{}
      node (ABCr) at +(4,0) [sommets]{}
      (as) edge[arcnt] node [midway,swap] {$A$} (at);
      

      \draw [arct] 
      (s) edge node [midway] {$a$} (t)
      (s) edge node [midway] {$a$} (ABCr)
      (ABq) edge node [midway] {$d$} (new1) 
      (ABq) edge node [midway,auto=right] {$d$} (s)
      (ABq) edge node [midway] {$a$} (Ap)
      (ABCr) edge[arcnt] node [midway] {$A$} (ABq);

    \end{tikzpicture}
  }
  \caption{An example of a graph grammar.}
\label{fig:gramrun}
\end{center}
\end{figure}
\end{example}

\begin{remark}
  Note that Definition~\ref{def:gragra} corresponds to the classical
  definition of deterministic hypergraph grammars
  \cite{courcelle,caucal:gragra}, since there is exactly one rewriting
  rule for each non-terminal symbol. Moreover, we implicitely assume
  that terminal symbols have arity one or two (Markov chains
  are transitions systems, thus arities greater than $2$ do not make
  sense in this context). This way, the generated graphs are coloured
  graphs (or transition systems where transitions and states are
  labelled).
  \end{remark}

  Let $\calG = (N, T, R, \axio)$ be a hypergraph grammar. Given $A\in
  N$ a non-terminal, we denote by $\versdans {R,A}$ the rewriting
  relation between hypergraphs with respect to the rule $(H_A,\iota_A) \in
  R$. Formally, a hypergraph $M$ rewrites into $M'$, written $M
  \versdans {R,A} M'$, if there exists a hyperarc $X = Av_1v_2\ldots
  v_p$ in $M$ such that $M' = (M- X) \cup h(H_A)$ where $h$ is an
  injective morphism that maps $\iota(i)$ to $v_i$ and other vertices
  of $H_A$ to vertices outside $M$. Intuitively, $M'$ is obtained from
  $M$ by replacing $X$ (of non-terminal label $A$) with $H_A$. The rewriting
  relation extends to the \emph{complete parallel rewriting} relation:
  the rewriting of each non-terminal simultenaously. We write $M
  \chemin{}{R} M'$ for the complete parallel rewriting of $M$ into
  $M'$. In other words, all non-terminal hyperedges of $M$ have been
  replaced in $M'$ using their respective rewriting rules in $R$.
  The set of all images of a graph $M$ by $\chemin{}{R}$ is
  denoted by $R[M]$. This set contains all isomorphic graphs obtained
  by applying the rules of $R$ to $M$.  For $n>1$, this notation is
  extended inductively into $R^n[M]=\bigcup_{M'\in R^{n-1}[M]}R[M']$,
  it is the set of all isomorphic graphs obtained after $n$
  applications of the complete parallel rewriting.

  Let $N$ and $T$ be sets of non-terminals, respectively
  terminals. Given $h$ a hypergraph labelled by $N \cup T$, we denote
  by $[H]$ the set of terminal arcs and colours in $H$: $[H] = H\cap
  (T_2\ V_H\ V_H\cup T_1\ V_H)$. For $\calG=(N, T, R, \axio)$ a \hrgg, the
  set of graphs generated by $\calG$ is defined as follows:
\[
\calG^\omega = \ens{\cup_{n\pgd0}[H_n]\ |\ H_0=Z \wedge \qqs n\pgd 0,
  H_n \chemin{}{R} H_{n+1} }
\]
Note that if $H_n \chemin{}{R} H_{n+1}$, then $[H_n] \subseteq
[H_{n+1}]$. Thus the set $\calG^\omega$ contains graphs which are all
isomorphic. A graph $H$ is \emph{generated} by $\calG$ if it belongs
to $\calG^\omega$.
Let $H\in \calG^\omega$, for each vertex $v\in V_H$, we let
$\level(v)$ be the \emph{level} at which $v$ is generated. Formally,
$\level(v) = \min \{k\, |\, v \in [H_k]\}$. Furthermore, notation
$\canon(v)$ stands for the {\em canonical} image of $v$ in the finite
set of vertices $\bigcup_{A\in N} V_{H_A}$. Assuming $H_{k-1}
\xrightarrow{R,A} H'_{k-1}$ for some $A\in N$ and $v\in H'_{k-1}$,
$\canon(v)$ is the unique vertex in $H_A$ whose image by $h$ is
$v$. When vertex $v$ is generated in $H_A$ at the $i$-th position of
an arc labelled by $B \in N$, we write $\canon(v) = (B,i)_A$.
Observe that, since $v\not\in H_{k-1}$, for each $j$, $v$ is distinct
from $\injec_A(j)$. 

\begin{example}\label{ex:rewrite}
  Figure~\ref{fig:gramrun} presents an example of a \hrgg,
  Figure~\ref{fig:rewrite} illustrates, starting from the axiom $Z$,
  two successive applications of the complete parallel rewriting
  (which coincides here with the rewriting of a single non-terminal)
  and the iteration of this process.
  In this example, each application of the rewriting rules adds new
  vertices as well as new arcs to the graph. Observe that the names of
  the vertices (except for $v_0$ that is distinguished) are not
  depicted, since they are not relevant to our purpose. Up to renaming
  of the vertices, there is a unique generated infinite graph.
  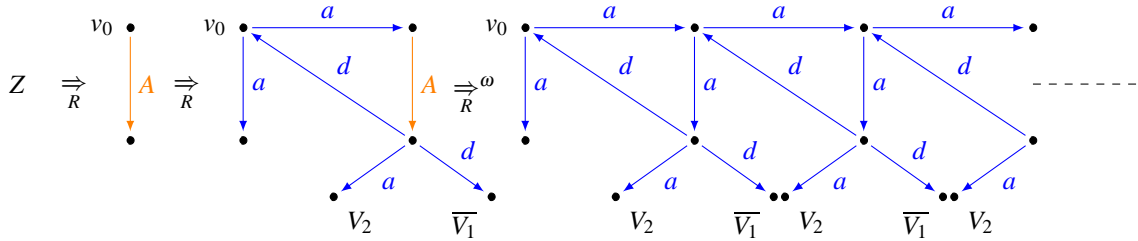
\begin{figure}[htbp]
    \begin{center}
      {\small
        \begin{tikzpicture}[scale=.75, auto=left]
          %
          
          \tikzstyle{every label}=[label distance= 0.5mm];
          \tikzstyle{sommets}=[draw,shape=circle,inner sep=0pt, outer sep = 0pt, fill= black,minimum size= 1mm];
          \tikzstyle{recriture}=[line width = 1pt, arrows = -latex];
          \tikzstyle{arct}=[arrows = -latex, shorten >=2pt, shorten <=2pt,color=blue];
          \tikzstyle{arcnt}=[arrows = -latex, shorten >=2pt, shorten <=2pt,color=orange];
      
          \draw 
          (-4,0) 
          node (Z) at +(0,1) {$Z$}
          (-2,0) 
          node (s0) at +(0,2) [sommets, label=left:$v_0$]{}
          node (t0) at +(0,0) [sommets]{} 
          (0,0) 
          node (s1) at +(0,2) [sommets, label=left:$v_0$]{}
          node (t1) at +(0,0) [sommets]{} 
          node (ABCr1) at +(3,2) [sommets]{} 
          node (ABq1) at +(3,0) [sommets]{} 
          node (Ap1) at +(1.6,-1) [sommets, label=below right:$V_2$]{}
          node (new11) at +(4.4,-1)[sommets, label=below left:$\overline {V_1}$]{}
      
          (5,0)
          node (s) at +(0,2) [sommets, label=left:$v_0$]{}
          node (t) at +(0,0) [sommets]{} 
          node (ABCr) at +(3,2) [sommets]{} 
          node (ABq) at +(3,0) [sommets]{} 
          node (Ap) at +(1.6,-1) [sommets, label=below right:$V_2$]{}
          node (new1) at +(4.4,-1)[sommets, label=below left:$\overline {V_1}$]{}
          node (ABCABCr) at +(6,2) [sommets]{} 
          node (ABCABq) at +(6,0) [sommets]{} 
          node (ABCAp) at +(4.6,-1) [sommets, label=below right:$V_2$]{}
          node (new2) at +(7.4,-1)[sommets, label=below left:$\overline {V_1}$]{}
          node (w) at +(9,2) [sommets]{} 
          node (v) at +(9,0) [sommets]{} 
          node (u) at +(7.6,-1) [sommets, label=below right:$V_2$]{}
          ++(9,1) -- +(1.9,0)[dashed, shorten >=5pt]
          (0,.8)
          node (e1) at +(-3,0){$ \chemin{}{R}$}
          node (e2) at +(-1,0){$ \chemin{}{R}$}
          node (e3) at +(4.1,0){$ {\chemin{}{R}}^\omega$}
          ;


          \draw [arct] 
          (ABq) edge node [midway] {$d$} (new1) 
          (ABq1) edge node [midway] {$d$} (new11) 
          (ABCABq) edge node [midway] {$d$} (new2)
          (s) edge node [midway] {$a$} (t)
          (s0) edge [arcnt] node [midway] {$A$} (t0)
          (s1) edge node [midway] {$a$} (t1)
          (ABq) edge node [midway] {$a$} (Ap)
          (ABq1) edge node [midway] {$a$} (Ap1)
          (s) edge node [midway] {$a$} (ABCr)
          (s1) edge node [midway] {$a$} (ABCr1)
          (ABq) edge node [midway,auto=right] {$d$} (s)
          (ABq1) edge node [midway,auto=right] {$d$} (s1)
          (ABCr) edge node [midway] {$a$} (ABq)
          (ABCr1) edge [arcnt] node [midway] {$A$} (ABq1)
          (ABCABq) edge node [midway,auto=right] {$d$} (ABCr)
          (ABCABq) edge node [midway] {$a$} (ABCAp)
          (ABCr) edge node [midway] {$a$} (ABCABCr)
          (ABCABCr) edge node [midway] {$a$} (ABCABq)
          (v) edge node [midway,auto=right] {$d$} (ABCABCr)
          (ABCABCr) edge node [midway] {$a$} (w)
          (v) edge node [midway] {$a$} (u);

        \end{tikzpicture}
      }
      \caption{Application of successive complete parallel rewritings and
        the  generated graph.}
      \label{fig:rewrite}
    \end{center}
  \end{figure}
\end{example}

\subsection{Basic Properties and Normal Forms for Regular Graphs}
\label{subsec:basics-hrgg}
For any rule $(H_A, \injec_A)$, we say that the vertices
$\injec_A(\{1,\cdots,\varrho(A)\})$ are the {\em inputs} of $H_A$, and
$\bigcup_{Y\in H_A \wedge Y(1)\in N_R}V_Y$ are the {\em outputs} of
$H_A$. In particular, output vertices belong to non-terminal
hyperedges.

Given a non-terminal $A\in N$, we denote by $\Succ(A)$ the set of
non-terminals appearing in $H_A$.

Given a \hrgg $\calG=(N, T, R, \axio)$ and a non-terminal hyperarc $X
= Av_1v_2\ldots v_p$, we introduce notations $R^\omega$
(resp. $R^\omega[X]$) to denote a particular graph in $\calG^\omega$
(resp. in $(\calG[X])^\omega$ with $\calG[X] := (N,T,R,X)$).

Let $\calG=(N, T, R, \axio)$ and $\calG'=(N', T', R', \axio')$ be two
\hrggs we say that $\calG'$ is a {\em colouring} of $\calG$ if, for
any graphs $H\in {\calG}^\omega$ and $H'\in {\calG'}^\omega$, there is
a graph isomorphism between $H$ and $H'$ which also preserves colours
of $H$, and there is a {\em colour} in $T'_1$ which does not belong to $T_1$.

We conclude these preliminaries by giving a normal form for
{\hrgg}s.

\begin{theorem}{\cite{caucal:gragra}}\label{theo:infinite}
  Any regular hypergraph can be generated in an effective way by a
  complete outside grammar.
\end{theorem}

The {\em complete outside} property ensures that the only input
vertices that are also outputs are vertices of infinite degree. It
also implies that each output vertex belongs to a single non-terminal
hyperarc. This property enables one to identify efficiently grammars
having vertices of infinite degree, and it also ensures that whenever
there is no such vertex, inputs and outputs are distinct.
In the sequel we assume that all {\hrgg}s we consider are complete
outside.

\subsection{Probabilistic Regular Graphs}

In order to obtain a probabilistic graph from one generated by a
\hrgg, we define, for each \hrgg $\calG$, and each graph $H$
in $\calG^\omega$, the counting function $\compte : V_H\times T_2\vers
\nit$, with $\compte(v,a)=\abs{\{v'\ |\ v\larc a v'\}}$, that
associates with each pair $(v,a)$ the number of $a$-labelled arcs
originating from $v$. Observe that two distinct vertices $v$ and $v'$
in $H$ have identical valuations for $\compte$ as soon as
$\canon(v)=\canon(v')$. 

\begin{definition}[Probabilistic graph grammar]
  A {\em probabilistic hypergraph grammar} (\phrgg for short) $\calP$,
  is a pair $(\calG,\proba)$ where $\calG =(N, T, R, \axio)$ is a \hrgg,
  $\proba: T_2\vers [0,1]$ is a mapping, and for each vertex $v\in
  R^\omega$ the sum of the $\proba$-values of all arcs from $v$
  is $1$: $\sum_{a\in T_2} \proba(a)\compte (v,a) = 1$.
\end{definition}

\begin{remark}
  This definition obviously precludes vertices with infinite
  out-degree.  In fact, it is not straightforward to introduce a
  meaningful definition enabling vertices having infinite out-degree.
  On the contrary, vertices with infinite in-degree are acceptable
  with this definition.
\end{remark}

\begin{proposition}
\label{prop:phrgg}
Given a \hrgg $\calG$ and a mapping $\proba: T_2\vers [0,1]$, one can
decide whether $(\calG, \proba)$ is a \phrgg.
\end{proposition}
\begin{proof}
  From Theorem~\ref{theo:infinite} we may assume that $\calG$ is
  complete outside. It enables to identify vertices of infinite
  out-degree. Let $v$ be such a vertex, and $a$ a label such that
  $\compte(v,a)=+\infty$, it forbids $(\calG,\mu)$ to be a \phrgg
  for any value of $\mu(a)$.
  If there is no such vertex, from Proposition~3.13~(b) of
  \cite{caucal:gragra}, there exists an effective colouring of
  $\calG=(N, T, R, \axio)$ with colours representing the degree of
  each vertex (relative to each label). We produce a colouring
  representing the exact {\em out}-degree relative to each element of
  $T_2$. There are only finitely many such degrees (from the same
  proposition, (a)). Now from these colours we are able to compute
  $\compte$ at each vertex $v$ in the grammar and therefore we may check
  that $\sum_{a\in T_2} \proba(a)\compte (v,a) = 1$.
\end{proof}

\begin{example}\label{ex:runningprob} 
  We consider the graph from Example~\ref{ex:running}.  The
  probabilistic mapping $\mu$, defined by $\mu(a)=\frac 1 2$ and
  $\mu(d)=\frac 1 4$, yields a probabilistic regular graph. Clearly
  the sum of out-going edges is $1$ for each vertex of the graph.
\end{example}

\subsection{Connection between regular graphs and pushdown automata}
\label{subsec:link}
There is a strong connection between regular graphs and configuration
graphs of pushdown automata. Indeed restricted to finite in- and
outdegrees, these graphs coincide: see, {\em
  e.g.},~\cite[Theorem~5.11]{caucal:gragra}. In particular, given a
pushdown automaton, the transformation into a graph grammar which
generates a infinite regular graph isomorphic to the configuration
graph of the pushdown system is straightforward and may be adapted
from the proof of Proposition~5.4 in~\cite{caucal:gragra}. This
proposition states that the suffix graph of any rewriting system may
be generated by a one rule grammar from the non-terminal. We
illustrate this construction on the following example.

\begin{example}\label{ex:pushdown}
  Let us consider the following pushdown system
  \[
    r \larc a Br' \quad  r' \larc a Ar \quad
    r' \larc b Ap \quad BAp \larc a p.
  \]
  To match more closely~\cite[Proposition~5.4]{caucal:gragra} it is
  presented as a suffix rewriting system: the state of the pushdown
  automaton is on the top of the stack, and rules are applied to
  suffixes of the stack. For example, when in state $r$, and whatever
  the contents of the stack, while reading an $a$, stack-symbol $B$ is
  pushed and the new state is $r'$. The transformation of this
  pushdown automaton into a graph grammar goes as follows. There is a
  unique non-terminal $X$ (which, hence, serves as axiom). The
  vertices of $H_X$ are words:
  each strict suffix (distinct from the empty
  suffix) of the words appearing in the rewriting rules (in the left-
  and right-hand sides) belongs to the image of $\injec_X$. Here $r$,
  $p$, $r'$ and $Ap$ are the non-empty strict suffixes and they are
  represented on the top line of the graph $H_X$. For every stack
  symbol (here $A$ and $B$), and every non-empty strict suffix, a
  vertex is formed by the concatenation of the stack symbol and the
  suffix. This yields new vertices, such as $Br$ and all the ones on
  the bottom line of $H_X$, but some vertices might already be
  present, as $Ap$ in this example. For each stack symbol, a
  non-terminal arc, labelled by $X$ connects these vertices:
  $Ar, Ap, Ar', AAp$ and $Br, Bp, Br', BAp$, respectively. This
  construction ensures that each left- and right-hand side of the
  rewriting rules is one vertex. It now suffices to add terminal
  arcs between the vertices according to the rules. For example the
  $a$-edge from $r$ to $Br'$ encodes the first rewriting rule. 
  
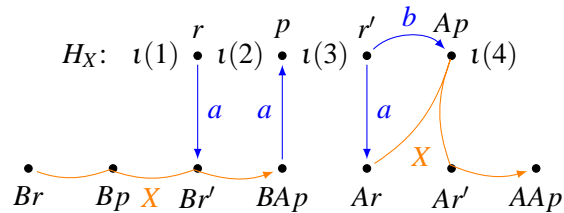
\begin{center}
    \begin{tikzpicture}[scale=.75, auto=left]
      \tikzstyle{every label}=[label distance= 0.5mm];
      \tikzstyle{sommets}=[draw,shape=circle,inner sep=0pt, outer sep = 0pt, fill= black,minimum size= 1mm];
      \tikzstyle{recriture}=[line width = 1pt, arrows = -latex];
      \tikzstyle{arct}=[arrows = -latex, shorten >=2pt, shorten <=2pt,color=blue];
      \tikzstyle{arcnt}=[arrows = -latex, shorten >=2pt, shorten <=2pt,color=orange];
      
      \draw 
      node (HX) at (0,2) {$H_X$:}
      ++(2,2) 
      node (r) at +(0,0) [sommets, label=above:$r$, label=left:$\injec(1)$]{}
      ++(1.5,0)
      node (p) at +(0,0) [sommets, label=above:$p$, label=left:$\injec(2)$]{} 
      ++(1.5,0)
      node (rp) at +(0,0) [sommets, label=above:$r'$, label=left:$\injec(3)$]{}
      ++(1.5,0)
      node (Ap) at +(0,0) [sommets, label=above:$Ap$, label=right:$\injec(4)$]{}
      ++(-7.5,-2) 
      node (Br) at +(0,0) [sommets, label=below:$Br$]{}
      ++(1.5,0)
      node (Bp) at +(0,0) [sommets, label=below:$Bp$]{}
      ++(1.5,0)
      node (Brp) at +(0,0) [sommets, label=below:$Br'$]{} 
      ++(1.5,0)
      node (BAp) at +(0,0) [sommets, label=below:$BAp$]{}
      ++(1.5,0)
      node (Ar) at +(0,0) [sommets, label=below:$Ar$]{}
      ++(1.5,0)
      node (Arp) at +(0,0) [sommets, label=below:$Ar'$]{} 
      ++(1.5,0)
      node (AAp) at +(0,0) [sommets, label=below:$AAp$]{}
      ;
      \draw [arct] 
      (r) edge node [midway] {$a$} (Brp)
      (rp) edge node [midway] {$a$} (Ar)
      (rp) edge [bend angle=45,bend left] node [midway] {$b$} (Ap)
      (BAp) edge node [midway] {$a$} (p)
      ;
      \foreach \deb/\posit in {A/near end,B/midway}
      \draw [arcnt, bend angle=20, bend right]
      (\deb r) to (\deb p) to node [\posit,auto=right]{$X$} (\deb rp) to (\deb Ap);
    \end{tikzpicture}
  \end{center}
  Notice that this construction produces several connected
  components. Yet, given an initial configuration only the connected
  component (co-)reachable from this configuration will be relevant.
\end{example}

A similar transformation can be applied to any pushdown automaton in
order to obtain a graph grammar which generates the configuration
graph of the pushdown system. This underlines the generality of the
model of graph grammars.
Moreover, we argue the framework of graph grammars is more convenient
than the pushdown automata view. Indeed,
transformations presented in Subsection~\ref{subsec:basics-hrgg} on
graph grammars do not affect the graph they generate, contrary to most
transformations on pushdown automata that affect the structure of the
configuration graph.

Esparza \emph{et al.} propose in~\cite{EKM06} a model of probabilistic
pushdown automata, derived from pushdown automata by assigning weights
to rules. The configuration graphs of such systems are infinite state
Markov chains. Probabilistic pushdown automata and \phrgg relate in
the same way than pushdown automata and graph grammars do:
the Markov chains defined by both models are the same. Moreover, any
probabilistic pushdown automaton can be turned into a \phrgg which
generated exactly the same infinite state Markov chain. In this sense
our model does not generalize the previous model. On the other hand,
\cite{EKM06} makes several syntactical assumptions on pushdown
automata which do not restrict the class of Markov chains, but make it
more difficult to manipulate. Transformations of probabilistic
pushdown automata in order to fit these assumptions may alter the
properties of the Markov chain. On the contrary, transformations of
\phrggs do not affect the Markov chain generated.

\section{Verification of \prgs}
\subsection{Markov chains and PCTL}

A (discrete-time) \emph{Markov chain} is a tuple $\mathcal{M} =
(S,s_0,p)$ consisting of a (possibly infinite) set $S$ of states, an
initial state $s_0$
, and a probabilistic transition function $p : S
\times S \rightarrow [0,1]$ such that for every state $s$, $\sum_{s'
  \in S} p(s,s') =1$. For simplicity, we assume the transition system
is finitely branching, \emph{i.e.}, in any state $s$ there are only
finitely many states $s'$ with $p(s,s')>0$; the condition $\sum_{s'
  \in S} p(s,s') =1$ is thus well-defined. Given a set of atomic
propositions $\AP$, a \emph{labelled Markov chain} $\mathcal{M} =
(S,s_0,p,\ell)$ is a Markov chain $(S,s_0,p)$ equipped with a
labelling function $\ell : S \rightarrow \AP$.

Introduced in~\cite{HJ-fac94}, PCTL is an extention of CTL with
probabilities. It can express quantitative properties about
executions in Markov chains, \emph{e.g.}, with probability $0.9$ any
sent message will be acknowledged in the future. The syntax of PCTL is
the following:
\[
\varphi ::= \texttt{tt} \ |\ a\ |\ \neg \varphi\ |\ \varphi \wedge
\psi\ |\ \Next^{\sim \rho} \varphi\ |\ \varphi \Until^{\sim \rho} \psi
\]
where $a \in \AP$ is an atomic proposition, $\rho \in [0,1]$ and $\sim
\in \{\leq,<,>,\geq\}$. Operators $\Next^{\sim \rho}$ and
$\Until^{\sim \rho}$ are respectively the probabilistic next-state and
until operators and generalise their nonprobabilistic
counterparts. Recall the shortcuts in CTL for eventually ($\F$) and
globally ($\G$): $\F \varphi \equiv \texttt{tt} \Until \varphi$ and $\G
\varphi \equiv \neg \F \neg \varphi$. Their probabilistic extensions
$\F^{\sim \rho}$ and $\G^{\sim \rho}$ will also be convenient in the
sequel.

Let $\mathcal{M} = (S,s_0,p,\ell)$ be a labelled Markov chain, and $s
\in S$. For a (non-probabilistic) formula $\phi$ of CTL, we write
$\prob(s \models \phi)$ for the measure of the set of paths in
$\mathcal{M}$ issued from $s$ and which satisfy $\phi$. Note that for
$V_1$ and $V_2$ sets of states, the set of paths from $s$ satisfying
$\Next V_1$ or $V_1 \Until V_2$ is clearly measurable. The semantics
of a PCTL formula $\varphi$ over $\mathcal{M}$ is defined inductively:

\noindent
\begin{center}
\begin{tabular}{ll}
  $\sem{{\texttt{tt}}} = S$ & \quad\quad $\sem{\varphi \wedge \psi} = \sem{\varphi} \cap \sem{\psi}$\\
  $\sem{a} = \{s \in S\ |\ a \in  \ell(s)\}$& \quad\quad $\sem{\Next^{\sim \rho} \varphi} = \{s \in S\ |\ \prob(s \models
  \Next \sem{\varphi}) \sim \rho\}$\\
  $\sem{\neg \varphi} = S \setminus \sem{\varphi}$ 
  & \quad\quad
  $\sem{\varphi \Until^{\sim \rho} \psi} = \{s \in S \ |\ \prob(s
  \models \sem{\varphi} \Until \sem{\psi}) \sim \rho\}$\\
$\sem{\F^{\sim \rho} \varphi} = \{s \in S\ |\ \prob(s \models \F
\sem{\varphi}) \sim \rho\}$ & \quad \quad $\sem{\G^{\sim \rho} \varphi} = \{s
\in S\ |\ \prob(s \models \G \sem{\varphi}) \sim \rho\}$
\end{tabular}
\end{center}

\noindent
and we write $s \models \varphi$ for $s \in \sem{\varphi}$.

In the following, we will interpret PCTL formulae over labelled Markov
chains induced by \phrgg. Atoms in these
formulae will be sets of vertices and will form the set of atomic
propositions $\AP$.
%
\begin{example}\label{ex:pctl}
  Considering the graph presented in Example~\ref{ex:running}, the
  probabilistic mapping given in Example~\ref{ex:runningprob}, and
  predicates $V_1$ and $V_2$ satisfied by vertices labelled by these
  respective colours, the following formulae are of interest:
  \begin{itemize}
  \item $\varphi_1 = V_1\wedge\Next^{\pgd \frac 1 2} V_2$: Vertices
    that satisfy $\varphi_1$ belong to $\conf1$ and with probability
    greater than $\frac{1}{2}$, their successors in one step are in
    $\conf2$. In particular, vertices at a fork on the lower line of
    Figure~\ref{fig:rewrite} satisfy $\varphi_1$.
  \item $\varphi_2 = v_0\wedge {V_1 \Until^{>\frac 2 3} V_2}$: Vertex
    $v_0$ satisfies $\varphi_2$ if the probability of all paths issued
    from $v_0$ that eventually reach $V_2$ passing through vertices of
    $V_1$ only is greater than $\frac{2}{3}$.
  \end{itemize}
\end{example}

\subsection{Qualitative model checking for \prgs}

The qualitative fragment of PCTL only involves the probability
thresholds $0$ and $1$.  Let $\calP= (N, T, R, \axio,\mu)$ be a
\phrgg. Up to isomorphism $\calP$ generates a unique infinite state
Markov chain $\M_\calP$ (or $\M$ when there is no ambiguity on
$\calP$).  The qualitative model checking problem for \prgs is, given
a \phrgg $\calP$ with initial vertex $v_0$ and a qualitative PCTL
formula $\varphi$, to answer whether in $\M_\calP$, $v_0 \models
\varphi$. Mimicking the finite Markov chain approach, the set of
vertices satisfying a qualitative formula can be effectively computed.

\begin{theorem}
\label{th:quali}
Let $\fifi$ be a qualitative PCTL formula, and $\calP$ a \phrgg. There
is an effective colouring $\calP'$ in which the set $\ens{v\in V_G\ |\
  v\models \fifi}$ is identified by a new colour.
\end{theorem}

\begin{proof}
  The proof is by induction on the structure of $\fifi$, using the
  fact that the following sets of vertices can be effectively coloured
  in the graph grammar: $\ens{v\in V_G\ |\ \prob(v, \Next{}
    \conf{})=1}$, $\ens{v\in V_G\ |\ \prob(v, \Next{} \conf{})=0}$,
  $\ens{v\in V_G\ |\ \prob(v, \conf{1} \Until{} \conf{2})=1}$ and
  $\ens{v\in V_G\ |\ \prob(v, \conf{1} \Until{} \conf{2})=0}$.

  Let us start with the two first cases: $\ens{v\in V_G\ |\ \prob(v,
    \Next{} \conf{})=1}$ and $\ens{v\in V_G\ |\ \prob(v, \Next{}
    \conf{})=0}$.
  The function $\canbis$ induces a finite partition on vertices of the
  infinite Markov chain generated by $\calP$. Two vertices with same
  image by $\canbis$ have equivalent successors. By hypothesis on the
  grammar, for every vertex generated at level $n$, all successor
  vertices are generated between levels $n-1$ and $n+1$. Hence, if $v$
  is generated in $H_A$, it is sufficient to identify in $R^2[A]$
  whether all successors of $v$ belong to $\conf{}$ or
  $\overline{\conf{}}$. One can thus, in the hypergraphs $H_A$ (for
  each $A \in N$), annotate by colours the vertices which have all
  their successors in $\conf{}$, as well as those which have no
  successors in $\conf{}$. These colours precisely correpond to the
  sets $\ens{v\in V_G\ |\ \prob(v, \Next{} \conf{})=1}$ and $\ens{v\in
    V_G\ |\ \prob(v, \Next{} \conf{})=0}$.

  The two other cases $\ens{v\in V_G\ |\ \prob(v, \conf1 \Until
    \conf2)=1}$ and $\ens{v\in V_G\ |\ \prob(v, \conf1 \Until
    \conf2)=0}$ are treated similarly. We detail here the colouring of
  $\ens{v\in V_G\ |\ \prob(v, \conf1 \Until \conf2)=1}$. 
For $B \in \Succ(A)$ and $i \leq \rho(B)$
we let $R((B,i)_A) = \{A_j |  \Dec(B_i,A_j) >0\}$. 
We then define inductively the sets:
\begin{itemize}
\item $W_0 = (H_\axio \cap \conf2) \cup \{v \,|\, \canon(v) =
  (B,i)_A \textrm{ and } \Win(B_i)_A =1 \}$, and
\item $W_{n+1} = W_n \cup \{v\, |\, \canon(v) = (B,i)_A \textrm{ and }
  \Win(B_i)_A + \sum_{A_j \in R((B,i)_A)} \Dec(B_i,A_j) =1 \textrm{ and }
   R((B,i)_A) \subseteq W_{n} \}$.
\end{itemize}
Vertices in $W_0$ are directly winning, either because they already
belong to $\conf2$ or because from $B_i$ in context $A$, the
probability to win without decreasing level is $1$. Vertices in
$W_{n+1}$ are also almost surely winning (\emph{i.e.} satisfy $\conf1
\Until \conf2$ with probability $1$) because they are winning without
decreasing level (factor $\Win(B_i)_A$) or firstly decreasing level and
then win from $A_j$ with probability $1$ (since $A_j \in W_n$).

Clearly, $\bigcup_{n=0}^\infty W_n = \ens{v\in V_G\ |\ \prob(v, \conf1
  \Until \conf2)=1}$ and the $W_n$'s can be iteratively computed and
annotated in the grammar by colours.
\end{proof}

\subsection{Probability computation for \prgs}
We now face the problem of computing, given $v_0$ an initial vertex in
$H_Z$ and $\phi$ a CTL formula, the probability in $\M_\calP$ of the
set of paths starting in $v_0$ and satisfying $\phi$:
$\prob_{\M_\calP}(v_0 \models \phi)$. This can be done inductively on
the structure of $\phi$, and the difficult part amounts to
computing, given $\conf1$ and $\conf2$ colours, the probability
starting in $v_0$ to satisfy $\conf1 \Until \conf2$, written
$\prob(v_0 \models \conf1 \Until \conf2)$. This subsection focuses on
solving this problem.

\subsubsection{Preliminaries and notations}
Without loss of generality we assume that vertices of $\conf1$ and
$\conf2$ are annotated in the grammar by colours (terminals of arity
$1$) and that $v_0$ appears in $H_Z$ the hypergraph of the rewriting
rule associated to the axiom $\axio$ of $\calP$. Using the levelwise
decomposition of the Markov chain $\mathcal{M}_{\mathcal{P}}$, we show
how to express $\prob(v_0 \models \conf1 \Until \conf2)$ as a solution
of a system of polynomial equations derived from the axiom and the
rules.

The hypotheses we demand on \phrggs ensure that the first step of any
path issued from a vertex of level $n$ either remains at level $n$ or
reaches one of the neighbour levels, $n-1$ and $n+1$ (from
Theorem~\ref{theo:infinite}, it corresponds to restricting to finite
degree). This fact will enable levelwise decomposition of paths in the
Markov chain.

To compute probabilities in Markov chains generated by \phrggs we
exploit the regularities of the underlying graphs. 
For $v$ a vertex of $\M_\calP$ with $\canon(v) \in H_A$, we write
$\M[v]$ for the part of $\M_\calP$ with underlying graph $R^\omega[A]$
which contains $v$ and no vertices of level $\level(v)
-1$. Intuitively, if $v$ has been generated by a non-terminal $A$, we
consider the infinite (sub-)Markov chain generated from this
non-terminal. For two vertices $v$ and $v'$ of $\M_\calP$ with
$\canon(v) = \canon(v') \in H_A$, the isomorphism of $\M[v]$ and
$\M[v']$ ensures that for any CTL formula $\phi$, $\prob_{\M[v]}(v
\models \phi) = \prob_{\M[v']}(v' \models \phi)$.  In particular, if
$\phi$ is the formula $\conf1 \Until \conf2$, we obtain that:
the probability to succeed satisfying $\conf1 \Until \conf2$ without
decreasing level is the same from $v$ and from $v'$. The probability
to satisfy $(\conf1 \setminus \conf2)$ while decreasing level of $1$
is also independent of the level, provided the initial state
corresponds to a fixed canonical representant $(B,i)_A$. This
motivates the introduction of notations for such probabilities, that
are determined by the context and are independent of the level.

Let $A,B \in N$ be non-terminals such that $B \in \Succ(A)$. Starting
in state $v$, with $\canbis(v) = (B,i)_A$, each successor state
belongs to $R^2[A]$, the sub-graph obtained from non-terminal $A$ by
two successive complete parallel rewritings. Given $i \leq
\rho(B)$ and $j \leq \rho(A)$ we introduce:
\begin{itemize}
\item $\Dec(B_i,A_j)$ as the probability from $v'$, with $\canbis(v') =
  (B,i)_A$, to reach $v$ such that $\level(v) = \level(v') -1$ and
  $v = \injec_A(j)$
  satisfying along the path:
  $(\conf1 \setminus \conf2) \cap (\level \geq \level(v'))$;
\item $\Win(B_i)_A$ as the probability from $v'$, with $\canbis(v') =
  (B,i)_A$, to fulfill $(\conf1 \cap \level \geq \level(v'))\Until
  \conf2$.
\end{itemize}
(Here $\level \geq k$ denotes that the current level is greater than a
 given natural $k$.)

As explained before, $\Dec(B_i,A_j)$ and $\Win(B_i)_A$ do not depend on
$v'$ and $v$ but only on their images by $\canbis$. Moreover,
$\Dec(B_i,A_j)$ expresses the probability to decrease level by one
while satisfying a given property and $\Win(B_i)_A$ is the probability
to win, \emph{i.e.}, to fulfill $\conf1 \Until \conf2$ without
decreasing level. This justifies the chosen notations. 

The levelwise decomposition of paths is given by vertices belonging
(when generated) to non-terminal. Thus, given $A,B,C,D \in N$ such
that $B,D \in \Succ(A)$ and $C \in \Succ(B)$, we introduce notations
for some probabilities that can be computed directly in any portion
$R^2[A]$ of the Markov chain.
\begin{itemize}
\item $p(B_i)_A$ is the probability in $R^2[A]$ from $v$ with
  $\canbis(v) = (B,i)_A$ to fulfill $\conf1 \Until \conf2$ without
  visiting any $v' = \injec_A(j)$ nor $v'$ with $\canbis(v') \in
  \{(C,k)_B,(B,h)_A\}$.
\item $p(B_i,D_h)_A$ is the probability in $R^2[A]$ from $v$ with
  $\canbis(v) = (B,i)_A$ to fulfill $\G (\conf1 \setminus \conf2)$ and
  reach $v'$ with $\canbis(v') =(D,j)_A$ before any $v''$ such that
  $\canbis(v'') \in \{(B,l)_A,(D,h')_A\}$ and $\level(v'') =
  \level(v)$.
\item $\overleftarrow{p}(B_i,A_j)$ is the probability in $R^2[A]$ from
  vertex $v$ with $\canbis(v) = (B,i)_A$ to reach $v'$ with
  $v' = \injec_A(j)$ and $\level(v') = \level(v)-1$ and satisfy 
  $\G(\conf1 \setminus \conf2)$ without seeing any $v'' \in
  \{(C,k)_B,(B,l)_A,(D,h)_A\}$.
\item $\overrightarrow{p}(B_i,C_k)_A$ is the probability in $R^2[A]$
  from $v$ with $\canbis(v) = (B,i)_A$ to reach $v'$ with $\canbis(v')
  = (C,k)_B$ satisfying $\G (\conf1 \setminus \conf2)$ without
  visiting any $v'' = \injec_A(j)$ nor $v'' \in \{(C,k')_B,(B,h)_A\}$.
\end{itemize}

Intuitively, there are several alternatives for paths starting in $v$
(with $\canbis(v) = (B,i)_A$) and for which $\conf1 \Until \conf2$ is
not falsified: either they satisfy $\conf1 \Until \conf2$ without
visiting any vertex at some position on a non-terminal hyperarc, or
they satisfy $\G \conf1 \setminus \conf2$ and reach some vertex $v'$
at a given position on a non-terminal hyperarc. The above
probabilities split these cases according the first $v'$ encountered:
$v'$ can be at the level of $v$ (at the $h$-th position in hyperarc
$D$), or at levels $n-1$ (thus of the form $\injec_A(j)$) or $n+1$ (at
the $k$-th position in hyperarc $C$).
As argued before, $p(B_i)_A$, $p(B_i,D_j)_A$,
$\overleftarrow{p}(B_i,A_j)$, and $\overrightarrow{p}(B_i,C_k)_A$ can
be computed directly in $R^2[A]$, obtained from $H_A$, $H_B$, and
$H_E$ for all $E\in Succ(A)\cup Succ(B)$.

\begin{example} 
  We compute these probabilities on
  Example~\ref{ex:running}: $p(A_2)_A = a$, $p(A_1,A_2)_A = a$,
  $\overleftarrow{p}(A_2,A_1) = d$ and $\overrightarrow{p}(A_1,A_2)_A=0$.
 \end{example}

\subsubsection{Computation of $\prob(v_0 \models \conf1 \Until \conf2)$}

\begin{theorem}
\label{th:least-sol}
The $\Dec(B_i,A_j)$'s and $\Win(B_i)_A$'s satisfy the following
equations:
\begin{eqnarray}
\label{eq:dec}
\Dec(B_i,A_j) = \overleftarrow{p}(B_i,A_j) + \sum_{D_h} p(B_i,D_h)_A \cdot \Dec(D_h,A_j)
+ \sum_{C_k} \overrightarrow{p}(B_i,C_k)_A \cdot \sum_{B_\ell} \Dec(C_k,B_\ell)
\cdot \Dec(B_\ell,A_j)\\
\label{eq:win}
\Win(B_i)_A = p(B_i)_A + \sum_{D_h} p(B_i,D_h)_A \cdot \Win(D_h)_A
+ \sum_{C_k} \overrightarrow{p}(B_i,C_k)_A \Bigl( \Win(C_k)_B +
\sum_{B_j} \Dec(C_k,B_j) \cdot \Win(B_j)_A\Bigr).
\end{eqnarray}
Moreover, if we add the following constraints:
\begin{itemize}
\item if $B_i \notin (\conf1 \setminus \conf2)$ then $\Dec(B_i,A_j)=0$ for
  every $A_j$, and
\item if $B_i \in \conf2$ then $\Win(B_i)_A=1$, and if $B_i \notin
  (\conf1 \cup \conf2)$ then $\Win(B_i)_A=0$;
\end{itemize}
the $\Dec(B_i,A_j)$'s and $\Win(B_i)_A$'s form the least solution of
this system of polynomial equations.
\end{theorem}

\begin{proof}
  The correctness of Equations~\ref{eq:dec} and~\ref{eq:win} is proved
  by partitioning the set of paths issued from vertex $v$ with
  $\canbis(v) = (B,i)_A$.

  Precisely, concerning Equation~\ref{eq:dec}, any path from $v$ with
  $\canbis(v) =(B,i)_A$ to $v' = \injec_A(j)$ 
  (and
  $\level(v') = \level(v) -1$) satisfying $\G \conf1 \setminus \conf2$
  falls in exactly one of the following cases:
\begin{itemize}
\item either it goes directly from $v$ to $v'$ without leaving $v$'s
  level;
\item or it reaches vertex $v''$ with $\canbis(v'') = (D,h)_A$ and
  $\level(v'') = \level(v)$, and then goes from $v''$ to $v'$;
\item or it reaches some vertex $v''$ with $\canbis(v'') = (C,k)_B$ and
  $\level(v'') = \level(v)+1$, and then returns to $v$'s level at
  vertex $v^{(3)}$ with $\canbis(v^{(3)}) = (B,\ell)_A$ and from there
  finally reaches $v'$.
\end{itemize}
This case distinction is illustrated on Figure~\ref{fig:schema} where
plain arrows represent paths in $R^2[A]$ (as presented earlier) and
dotted arrows represent recursive probabilities to decrease level.
%
\begin{figure}[htbp]
\begin{center}
  {\small
    \begin{tikzpicture}[scale=.75, auto=left]
      %
      
      \tikzstyle{every label}=[label distance= 0.5mm];
      \tikzstyle{sommets}=[draw,shape=circle,inner sep=0pt, outer sep = 0pt, fill= black,minimum size= 1mm];
      \tikzstyle{recriture}=[line width = 1pt, arrows = -latex];
      \tikzstyle{arcnt}=[arrows = -latex, shorten >=2pt, shorten <=2pt,color=orange];
      \tikzstyle{direct}=[arrows = -latex, shorten >=2pt, shorten <=2pt];
      \tikzstyle{indirect}=[arrows = -latex, dashed, shorten >=2pt, shorten <=2pt];
      \tikzstyle{surligne}=[line width = 3pt, color=blue!10];
      
      \draw 
      (0,2) 
      +(0,1.3) node [arcnt](laba) {$A$}
      +(0,-.3) node (aj) {$A_j$}
      +(-.7,1) rectangle +(.7,-1.5)
      ++(5,2.5) 
      +(0,1.3) node [arcnt](labb) {$B$}
      +(0,0) node (bi) {$B_i$}
      +(0,-2) node[fill=blue!10] {$B_\ell$}
      node (bl) {$B_\ell$}
      +(-.7,1) rectangle +(.7,-3)
      ++(0,-5) 
      +(0,1.3) node [arcnt](labd) {$D$}
      +(0,-.3) node[fill=green!20]  {$D_h$}
       node (dh) {$D_h$}
      +(-.7,1) rectangle +(.7,-1.5)
      ++(5,5) 
      +(0,1.3) node [arcnt](labc) {$C$}
      +(0,-.3) node[fill=blue!10] {$C_k$}
      node (ck) {$C_k$}
      +(-.7,1) rectangle +(.7,-1.5)
      ;

      \draw [surligne]
      (bi) -- (ck) -- (bl) -- (aj);
      \draw [surligne,color=green!20]
      (bi) edge [bend left] (dh) (dh) -- (aj);

      \draw [direct] 
      (bi) edge node [pos=.3, auto=right] {$\overleftarrow{p}(B_i,A_j)$} (aj)
      (bi) edge node [pos=.3, auto=left] {$\overrightarrow{p}(B_i,C_k)$} (ck)
      (bi) edge [bend left] node [pos=.5, auto=left] {$p(B_i,D_h)$} (dh)
      ;
      \draw [indirect] 
      (ck) edge node [pos=.7, auto=left] {$\Dec(C_k,B_\ell)$} (bl)
      (dh) edge node [pos=.4, auto=left] {$\Dec(D_h,A_j)$} (aj)
      (bl) edge node [pos=.7, auto=left] {$\Dec(B_\ell,A_j)$} (aj)
      ;
      \draw 
      (-1,-.8)  node [label=30:$H_A$]{}-- ++(0,5.1) -- ++(7,3) -- ++(0,-10) -- (-1,-.8);
      \draw 
      (4,1.2) -- ++(0,4.6) -- ++(7,1.5) -- ++(0,-7.4) node [label=120:$H_B$]{} -- (4,1.2);
    \end{tikzpicture}
  }
  \caption{Illustration of Equation~\eqref{eq:dec} for $\Dec(B_i,A_j)$.}
\label{fig:schema}
\end{center}
\end{figure}
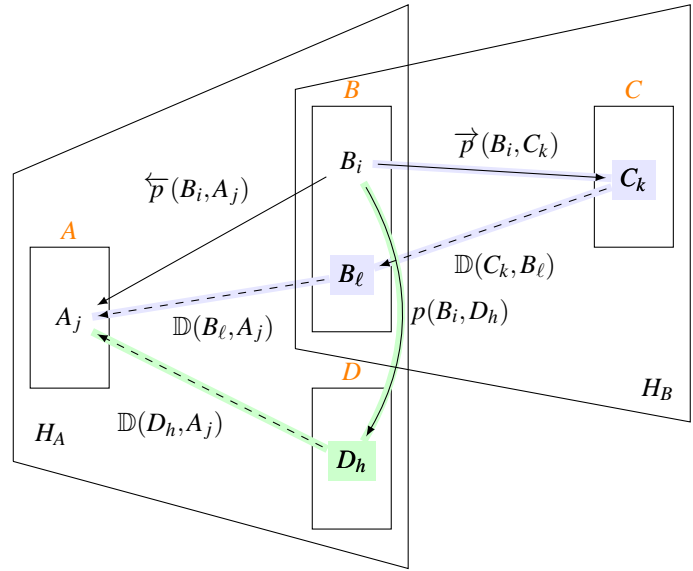

For Equation~\ref{eq:win}, the reasoning is similar. Any path issued
from $v$ satisfying $\conf1 \Until \conf2$ without visiting vertices
of level smaller than $\level(v)$:
\begin{itemize}
\item either satisfies $\conf1 \Until \conf2$ without visiting any
  other non-terminals (and hence at $v$'s level)
\item or reaches a vertex $v'$ with $\canbis(v') = (D,h)_A$ and
  $\level(v') = \level(v)$ and from then on satisfies $\conf1 \Until
  \conf2$ without decreasing level
\item or goes to vertex $v'$ with $\canbis(v') = (C,k)_B$ and
  $\level(v') = \level(v) +1$, and from there either satisfies $\conf1
  \Until \conf2$ without going back to verticesat $v$'s level, or
  reaches some $v''$ with $\canbis(v'') = (B,\ell)_A$ and $\level(v'') =
  \level(v)$ and from $v''$ satisfy $\conf1 \Until \conf2$ without
  decreasing level.
\end{itemize}
These partitions of the set of paths issued from vertex $v$ with
$\canbis(v) = (B,i)_A$ justify Equations~\ref{eq:dec} and~\ref{eq:win}.

The system of equations defines an operator $\mathcal{F} : [0,1]^n
\rightarrow [0,1]^n$ where $n$ is the number of variables appearing in
the system. The valuation $\mathcal{F}(\nu)$ of the variables is
obtained by evaluating each equation the right-hand side where each
variable is substituted with its value in $\nu$. This operator is
monotonic and continuous, and hence admits a unique least fixed-point,
which is eventually reached by iterating $\mathcal{F}$ on the
null-valuation which assigns $0$ to all variables. Note that the
convergence towards the least fixed-point might require infinitely
many iterations.

To prove that the $\Dec(B_i,A_j)$'s and $\Win(B_i)$'s form the {\em
  least} solution of the system, we consider the probabilities
approximated by truncating the paths at length $k$. Precisely, let
$\Dec(B_i,A_j)^k$ be the probability-mass of $\Dec(B_i,A_j)$
restricted to paths of length at most $k$, ; similarly let
$\Win(B_i)_A^k$ be the probability-mass of paths of length at most $k$
in $\Win(B_i)_A$. As $k$ tends to infinity, those probabilities tend
to $\Dec(B_i,A_j)$ and $\Win(B_i)_A$, respectively. It is thus
sufficient to prove that, for any $k \in \mathbb{N}$,
$\Dec(B_i,A_j)^k$ and $\Win(B_j)_A$ are no greater than the least
solution of the system. This is easily done by induction on $k$.
\end{proof}

Recall that our goal is to compute $\prob(v_0 \models \conf1 \Until
\conf2)$. This probability can be expressed using the
$\Dec(B_i,A_j)$'s and $\Win(B_i)_A$'s:
\begin{equation}
\label{eq:proba}
\prob(v_0 \models \conf1 \Until \conf2) = p(v_0)_\axio\ + 
\sum_{A_i \in \Succ(\axio)} \overrightarrow{p}(v_0,A_i)_\axio\
\Win(A_i)_\axio,
\end{equation}
where
\begin{itemize}
\item $p(v_0)_\axio$ is the probability in $H_\axio$ from $v_0$ to
  fulfill $\conf1 \Until \conf2$ without visiting any vertex $v'$ with
  $\canbis(v') = (A,i)_\axio$ for some $A \in Succ(\axio)$;
\item $\overrightarrow{p}(v_0,A_i)_\axio$ is the probability in $H_Z$
  from $v_0$ to $v'$ with $\canbis(v') = (A,i)_\axio$ while satisfying
  $\G (\conf1 \setminus \conf2)$ and without visiting any vertex $v''$
  such that $\canbis(v'')= (B,j)_\axio$ (for some $B \in
  \Succ(\axio)$) in between.
\end{itemize}

\begin{example}
  We illustrate the computation of $\prob(v_0 \models \conf1 \Until
  \conf2)$ on our running example. 
  Since $\canbis(v_0) = (A,1)_\axio$ and $v_0 \notin \conf2$, $p(v_0)_\axio
  =0$ and $\overrightarrow{p}(v_0,A_1) =1$. From
  Equation~\ref{eq:proba} we deduce $\prob(v_0 \models \conf1 \Until
  \conf2) = \Win(A_1)_\axio$. Let us detail some steps of the
  computation.
  \begin{align*}
 \Win(A_1)_\axio &= a \Win(A_2)_\axio + a
  \bigl(\Win(A_1)_A + \Dec(A_1,A_1) \Win(A_1)_\axio + \Dec(A_1,A_2)
  \Win(A_2)_\axio \bigr) \\
  &= a \Win(A_1)_A + a \Dec(A_1,A_1) \Win(A_1)_Z,
 \end{align*}
 since $\Win(A_2)_Z =0$. The probability $\Win(A_2)_A$ is easily
 computed: $\Win(A_2)_A =a$. Then $\Dec(A_1,A_1)$ is the least
 solution of a quadratic equation:
\[
a \Dec(A_1,A_1)^2 - \Dec(A_1,A_1) + a d = 0.
\]
Letting that $a=\frac{1}{2}$ and $d= \frac{1}{4}$, we get
$\Dec(A_1,A_1) = 1 - \frac{\sqrt 3}{2}$. Finally
\begin{align*}
\Win(A_1)_\axio &= \frac{a \Win(A_1)_A}{1 - a \Dec(A_1,A_1)}
= \frac{a^3}{(1-a -a \Dec(A_1,A_1))(1- a \Dec(A_1,A_1))}
\Win(A_1)_\axio= \frac{2}{3} (2\sqrt{3} -3) \approx 0.31.
\end{align*}
\end{example}

Note that the exact computation of the solutions of the system may not
always be performed. Indeed, in general, the equations are polynomials
(of arbitrary degree) in the variables. However, similarly as in
\cite{EKM06}, approximate values for the solutions can be
computed.

\begin{theorem}\label{theo:approx}
  Let $\calP = (N,T,R,Z,\mu)$ be a \phrgg, and $v_0$ a vertex in
  $H_\axio$. For $\rho\in \ratio\cap [0,1]$ and $\sim \in \ens{\pptit,
    <, \pgd, >}$, it is decidable whether $\prob(v_0 \models \conf1
  \Until \conf2) \sim \rho$. Moreover, given $0<\lambda<1$, one can
  compute $\rho_1, \rho_2 \in \ratio$ such that $\rho_1\leq\prob(v_0
  \models \conf1 \Until \conf2)\leq\rho_2$, and $\rho_2-\rho_1\pptit
  \lambda$.
\end{theorem}

\begin{proof} 
  Deciding $\prob(v_0 \models \conf1 \Until \conf2)\sim \rho$ is
  equivalent to deciding $p(v_0)_\axio\ + \sum_{A_i \in \Succ(\axio)}
  \overrightarrow{p}(v_0,A_i)_\axio \Win(Ai)_\axio\sim \rho$. Using
  Equations~\ref{eq:dec} and~\ref{eq:win}, the decidability of the
  first order arithmetics of reals~\cite{tarski51} yields the
  decidability of our problem.  An iterative application of the
  decision algorithm allows to compute in a dichotomic way the desired
  approximations $\rho_1$ and $\rho_2$.
\end{proof}

\subsection{Undecidability of quantitative model checking}

In this subsection, we give a proof of the undecidability of the exact
quantitative PCTL model-checking problem for {\phrgg}s. Since
{\phrgg}s generalise probabilistic pushdown automata, this result is a
consequence of the undecidability of quantitative PCTL model-checking
for probabilistic pushdown automata~\cite{BKS-stacs05}. We however
adapt the proof presented in \cite{BKS-stacs05} to graph grammars,
which, in our opinion, enable a simpler exposition.

The undecidability is proved by a reduction of Post Correspondance
Problem (PCP).  Recall that an instance of the PCP is a sequence of
pairs of words $((u_i,v_i))_{i\leq n}$ over a fixed alphabet
$\alphab$, and the problem is to determine whether there is an integer
$k$, and a sequence $(i_\ell) _{\ell\leq k}$ such that
$u_{i_1}u_{i_2}\ldots u_{i_{k}}=v_{i_1}v_{i_2}\ldots v_{i_k}$.

The quantitative model-checking problem of PCTL for {\phrgg}s is the
following:
\begin{center}
\begin{minipage}{0.8\textwidth}
  \noindent {Instance:} A \phrgg $\calP$, 
  and a PCTL formula $\varphi$.

\noindent {Question:} Is $\varphi$ valid on $\M_\calP$?
\end{minipage}
\end{center}

\begin{theorem}[\cite{BKS-stacs05}]
\label{th:undec}
The quantitative model-checking problem of PCTL for {\phrgg}s is
undecidable.
\end{theorem}

\begin{proof}
  This result is a consequence of~\cite{BKS-stacs05} but we give here
  a direct proof. Let $((u_i,v_i))_{i\leq n}$ be a sequence of pairs
  of words on $\alphab=\ens{0,1}$.  From this instance of PCP, we
  define the following \phrgg: $\mathcal{P}=(N,T,R,Z, \mu)$, where:
\begin{itemize}
\item $N=\ens{Z}_0\cup \ens{\new_i\ |\ i\leq n}_2$;
\item $T = \ens{s, green, red}_1\cup\ens{a,b}_2$;
\item $ \mu(a)=0.5, \mu(b)=1$;
\end{itemize}
and the set $R = (H_B,\injec_B)_{B\in N}$ of rewriting rules is
depicted below:
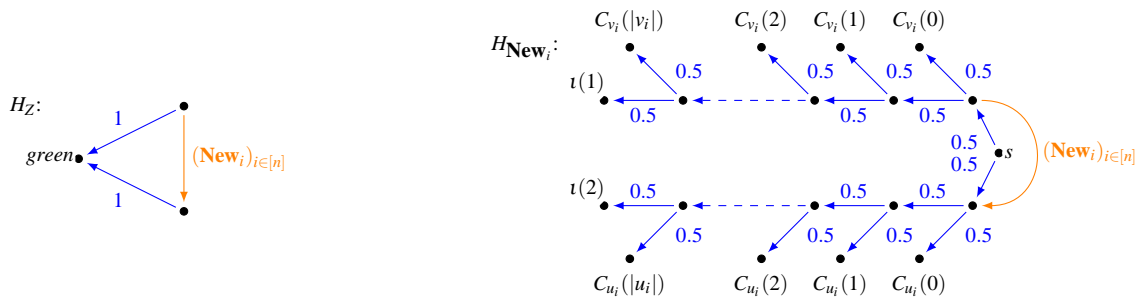
\begin{center}
  \begin{minipage}{0.35\linewidth}
    {\scriptsize
      \begin{tikzpicture}[scale=.7]
        \tikzstyle{every label}=[label distance= 0.5mm];
        
        \tikzstyle{sommets}=[draw,shape=circle,inner sep=0pt, outer
                        sep = 0pt, fill= black, draw=black,minimum size= 1mm];

        \tikzstyle{arct}=[arrows = -latex, shorten >=2pt, shorten <=2pt,color=blue, bend angle=20];
        \tikzstyle{arcnt}=[arrows = -latex, shorten >=2pt, shorten <=2pt,color=orange, bend angle=20];
        \tikzstyle{recriture}=[line width = 1pt, arrows = -latex,auto=left];
        
        \node at (-1,2) {$H_Z$:};

        \begin{scope}[every node/.style =  {sommets}, every label/.style ={draw=white, fill=white}]
          \draw (0,1) node (orig)[label=left:$green$]{}
          +(2,-1) node (bas){}
          +(2,1) node (haut){}
          ;
        \end{scope}

        \draw[arct]
        (bas) edge node [midway, auto=left] {$1$} (orig)
        (haut) edge node [midway, auto=right] {$1$} (orig)
        (haut) edge [arcnt] node [midway, auto=left] {$(\new_i)_{i\in [n]}$} (bas)
        ;
      \end{tikzpicture}
    }
  \end{minipage}
  \hfill
  \begin{minipage}{0.6\linewidth}
    {\scriptsize
      \begin{tikzpicture}[scale=.7, auto=left]
        \tikzstyle{every label}=[label distance= 0.5mm];
        
        \tikzstyle{sommets}=[draw,shape=circle,inner sep=0pt, outer
                sep = 0pt, fill= black, draw=black,minimum size= 1mm];
        
        \tikzstyle{recriture}=[line width = 1pt, arrows = -latex];

        \tikzstyle{arct}=[arrows = -latex, shorten >=2pt, shorten <=2pt,color=blue, bend angle=20];
        \tikzstyle{arcnt}=[arrows = -latex, shorten >=2pt, shorten <=2pt,color=orange, bend angle=20];
        
        \node at (-1.5,3) {$H_{\new_i}$:};
        
        \begin{scope}[every node/.style =  {sommets}, every label/.style ={draw=white, fill=white}]
          \draw (0,0) node (but01)[label=above left:$\injec(2)$]{}
          +(0,2) node (but21)[label=above left:$\injec(1)$]{}
          \foreach \niveau in {0,2}{
            \foreach \posit/\nomit in {1.5/2, 4/3, 5.5/4, 7/5}{
              +(\posit, \niveau) node (but\niveau\nomit){}
            }
          }
          +(7.5,1) node (debut)[label=right:$s$]{} 
          ;
        \end{scope}
        \foreach \niveau/\uivi/\posbel/\posit in {0/u_i/below/below left of,2/v_i/above/above left of}{
          \foreach \nomit/\indice in {2/\abs{\uivi},3/2,4/1,5/0}{
            \node[sommets,node distance=1cm] (target\niveau\nomit)
            [\posit=but\niveau\nomit,label=\posbel:$C_{\uivi}(\indice)$]  {};
          }
        }
        \draw[arcnt] 
        (but25) edge [bend angle=90, looseness=1.9,bend left] node [midway] {$(\new_i)_{i\in[n]}$} (but05)
        ;
        \draw[arct]
        \foreach \niveau/\cote/\acote in {0/right/left, 2/left/right}{
          (debut) edge node [midway, auto=\cote] {$0.5$} (but\niveau5)
          (but\niveau2) edge node [midway, auto=\cote] {$0.5$} (but\niveau1)
          (but\niveau3) edge [dashed,auto=\cote] node [midway] {} (but\niveau2)
          (but\niveau2) edge node [near start, auto=\acote] {$0.5$} (target\niveau2)
          \foreach \nomit/\nomplus in {3/4,4/5}{%
            (but\niveau\nomplus) edge node [midway, auto=\cote] {$0.5$} (but\niveau\nomit)
            (but\niveau\nomit) edge node [near start, auto=\acote] {$0.5$} (target\niveau\nomit)
          }
          (but\niveau5) edge node [near start, auto=\acote] {$0.5$} (target\niveau5)
        }
        ;
      \end{tikzpicture}
    }
  \end{minipage}
\end{center}
Colours $green$, and $red$ label vertices as follows. For each $i\leq
n$, $k\leq \abs{u_i}$, and $k'\leq \abs{v_i}$,
 \[
   C_{u_i}(k)=\left\{
     \begin{array}{cc}
       green & \mbox{ if } u_i(k) = 1 \\
       red & \mbox{ if } u_i(k) = 0
     \end{array}
   \right. , 
  \quad C_{v_i}(k')=\left\{
     \begin{array}{cc}
       green & \mbox{ if } v_i(k')=0 \\
       red & \mbox{ if } v_i(k')=1
     \end{array}
   \right. 
 \]
     
Consider the following PCTL formula:
  \[
   \laformule= S \wedge ( \texttt{tt}\Until^{=\frac 1 2} Green)
  \]
  where $Green$ and $S$ are atomic propositions corresponding to
  vertices labelled respectively by $green$ and $s$, terminals of arity
  $1$. We claim that $\laformule$ is valid on
  $\M_{\mathcal{P}}$ if and only if there is a solution to
  the Post instance $((u_i,v_i))_{i\leq n}$.

  In the infinite graph generated by $\calP$,
  each vertex labelled $s$ is connected to the origin (labelled $green$
  in $H_Z$) via a sequence of $u_i$'s on the lower branch, and of
  $v_i$'s on the upper branch (with the same indices). 
  \newcommand{\upath}{u_{\mathsf{path}}} 
  \newcommand{\vpath}{v_{\mathsf{path}}}
  Let $I=(I_0, I_1, \cdots, I_m)$ be a sequence of indices in
  $\{1,\cdots,n\}$, and consider $v_I$ the $s$-vertex corresponding to
  this sequence.  The probability to reach $red$ from $v_I$ is the
  following: $ \prob(v_I\models \texttt{tt}\Until
  Red)=\frac 12 (\upath + \vpath) $ with
\begin{equation*}
  \upath = \sum_{j\leq m}\ \sum_{k\leq \abs{u_j}}\frac 1{2^{(\sum_{\ell<j}\abs{u_{I_\ell}})+k}}(u_j(k) = 0) \quad \textrm{and}\quad
  \vpath = \sum_{j\leq m}\ \sum_{k'\leq \abs{v_j}}\frac 1{2^{(\sum_{\ell<j}\abs{v_{I_\ell}})+k}}(v_j(k')=1).
\end{equation*}
The only situation where $\prob(v_I\models \texttt{tt} \Until Red)=
\frac{1}{2}$ (and hence $\prob(v_I\models \texttt{tt} \Until Green) =
\frac{1}{2}$) occurs when the same sequence of letters appear in
$\upath$ and $\vpath$ (from the unicity of the binary expansion).
\end{proof}

\section {Conclusion}

In this paper we introduced \prgs, as graphs generated by graph
grammars where terminal arcs are labelled with probabilities. Results
concerning the model-checking of probabilistic pushdown automata
extend to this context. Precisely, both the approximate PCTL and
qualitative PCTL model checking problems are decidable, whereas the
exact quantitative model-checking problem is undecidable. 

We believe that our model of \phrggs offers a major benefit compared
to pushdown systems: it focuses on structural aspects whereas
configurations graphs of pushdown automata emphasise combinatorial
aspects. Furthermore in order to identify classes of infinite state
systems with a decidable quantitative PCTL model checking we believe
that {\em structural} restrictions on the grammar might prove worth
studying. A natural extension of our work is to extend the positive
results to graphs where infinite in-degree in allowed. Another
research direction is to try to climb up the Caucal hierarchy,
like~\cite{cara_worh
}, and pursue our work on higher-order pushdown systems.

\bibliographystyle{eptcs}

\begin{thebibliography}{10}
\providecommand{\bibitemstart}[1]{\bibitem{#1}}
\providecommand{\bibitemend}{}
\providecommand{\bibliographystart}{}
\providecommand{\bibliographyend}{}
\providecommand{\url}[1]{\texttt{#1}}
\providecommand{\urlprefix}{Available at }
\providecommand{\bibinfo}[2]{#2}
\bibliographystart

\bibitemstart{BBHK-lpar08}
\bibinfo{author}{T.~Br{\'a}zdil}, \bibinfo{author}{V.~Brozek},
  \bibinfo{author}{J.~Holecek} \& \bibinfo{author}{A.~Ku\v{c}era}
  (\bibinfo{year}{2008}): \emph{\bibinfo{title}{Discounted Properties of
  Probabilistic Pushdown Automata}}.
\newblock In: {\sl \bibinfo{booktitle}{Proceedings of the 15th International
  Conference on Logic for Programming, Artificial Intelligence, and Reasoning
  (LPAR'08)}}, {\sl \bibinfo{series}{Lecture Notes in Computer Science}}
  \bibinfo{volume}{5330}, \bibinfo{publisher}{Springer}, pp.
  \bibinfo{pages}{230--242}.
\bibitemend

\bibitemstart{BKS-stacs05}
\bibinfo{author}{T.~Br{\'a}zdil}, \bibinfo{author}{A.~Ku\v{c}era} \&
  \bibinfo{author}{O.~Strazovsk{\'y}} (\bibinfo{year}{2005}):
  \emph{\bibinfo{title}{On the Decidability of Temporal Properties of
  Probabilistic Pushdown Automata}}.
\newblock In: {\sl \bibinfo{booktitle}{Proceedings of the 22nd Annual Symposium
  on Theoretical Aspects of Computer Science (STACS'05)}}, {\sl
  \bibinfo{series}{Lecture Notes in Computer Science}} \bibinfo{volume}{3404},
  \bibinfo{publisher}{Springer}, pp. \bibinfo{pages}{145--157}.
\bibitemend

\bibitemstart{buchi:regular}
\bibinfo{author}{J.~R. B{\"u}chi} (\bibinfo{year}{1964}):
  \emph{\bibinfo{title}{Regular Canonical Systems}}.
\newblock {\sl \bibinfo{journal}{Archiv für Mathematische Logik und
  Grundlagenforshung}} \bibinfo{volume}{6}, pp. \bibinfo{pages}{91--111}.
\bibitemend

\bibitemstart{cara_worh}
\bibinfo{author}{A.~Carayol} \& \bibinfo{author}{S.~Woerhle}
  (\bibinfo{year}{2003}): \emph{\bibinfo{title}{The Caucal Hierarchy of
  Infinite Graphs in Terms of Logic and Higher-Order Pushdown Automata}}.
\newblock In: {\sl \bibinfo{booktitle}{Proceedings of the 23rd Conference on
  Foundations of Software Technology and Theoretical Computer Science
  (FSTTCS'03)}}, {\sl \bibinfo{series}{Lecture Notes in Computer Science}}
  \bibinfo{volume}{2914}, \bibinfo{publisher}{Springer}, pp.
  \bibinfo{pages}{112--123}.
\bibitemend

\bibitemstart{caucal:gragra}
\bibinfo{author}{D.~Caucal} (\bibinfo{year}{2007}):
  \emph{\bibinfo{title}{Deterministic graph grammars}}, {\sl
  \bibinfo{series}{Texts in logics and games}}~\bibinfo{volume}{2}, pp.
  \bibinfo{pages}{169--250}.
\newblock \bibinfo{publisher}{Amsterdam University Press}.
\bibitemend

\bibitemstart{caucal01}
\bibinfo{author}{D.~Caucal} \& \bibinfo{author}{T.~Knapik}
  (\bibinfo{year}{2001}): \emph{\bibinfo{title}{An internal presentation of
  regular graphs by prefix-recognizable ones}}.
\newblock {\sl \bibinfo{journal}{Theory of Computing Systems}}
  \bibinfo{volume}{34}(\bibinfo{number}{4}).
\bibitemend

\bibitemstart{courcelle}
\bibinfo{author}{B.~Courcelle} (\bibinfo{year}{1990}):
  \emph{\bibinfo{title}{Graph rewriting: an algebraic and logic approach}},
  {\sl \bibinfo{series}{Handbook of Theoretical Computer Science}}
  \bibinfo{volume}{B: Formal Models and Semantics}, pp.
  \bibinfo{pages}{193--242}.
\newblock \bibinfo{publisher}{Elsevier}.
\bibitemend

\bibitemstart{EKM06}
\bibinfo{author}{J.~Esparza}, \bibinfo{author}{A.~Ku{\v{c}}era} \&
  \bibinfo{author}{R.~Mayr} (\bibinfo{year}{2006}): \emph{\bibinfo{title}{Model
  Checking Probabilistic Pushdown Automata}}.
\newblock {\sl \bibinfo{journal}{Logical Methods in Computer Science}}
  \bibinfo{volume}{2}(\bibinfo{number}{1}).
\bibitemend

\bibitemstart{HJ-fac94}
\bibinfo{author}{H.~Hansson} \& \bibinfo{author}{B.~Jonsson}
  (\bibinfo{year}{1994}): \emph{\bibinfo{title}{A logic for reasoning about
  time and reliability}}.
\newblock {\sl \bibinfo{journal}{Formal Aspects of Computing}}
  \bibinfo{volume}{6}(\bibinfo{number}{5}), pp. \bibinfo{pages}{512--535}.
\bibitemend

\bibitemstart{Kucera-entcs06}
\bibinfo{author}{A.~Ku\v{c}era} (\bibinfo{year}{2006}):
  \emph{\bibinfo{title}{Methods for Quantitative Analysis of Probabilistic
  Pushdown Automata}}.
\newblock {\sl \bibinfo{journal}{Electronic Notes in Theoretical Computer
  Science}} \bibinfo{volume}{149}(\bibinfo{number}{1}), pp.
  \bibinfo{pages}{3--15}.
\bibitemend

\bibitemstart{muller}
\bibinfo{author}{D.~Muller} \& \bibinfo{author}{P.~Schupp}
  (\bibinfo{year}{1985}): \emph{\bibinfo{title}{The theory of ends, pushdown
  automata, and second-order logic}}.
\newblock {\sl \bibinfo{journal}{Theoretical Computer Science}}
  \bibinfo{volume}{37}, pp. \bibinfo{pages}{51--75}.
\bibitemend

\bibitemstart{tarski51}
\bibinfo{author}{A.~Tarski} (\bibinfo{year}{1951}): \emph{\bibinfo{title}{A
  Decision Method for Elementary Algebra and Geometry}}.
\newblock \bibinfo{publisher}{University of California Press, Berkeley}.
\bibitemend

\bibliographyend
\end{thebibliography}

\end{document}